\documentclass[12pt,pra,reprint,superscriptaddress,amsmath,amsfonts,amssymb]{revtex4-1}

\usepackage[pretty]{revquantum}
\usepackage[english]{babel}
\usepackage{graphicx}
\usepackage[caption=false]{subfig}
\usepackage{mathrsfs}
\usepackage{amsmath}
\usepackage{amsthm}
\usepackage{thmtools, thm-restate}
\usepackage{bm}
\usepackage{xcolor}
\usepackage{hyperref}
\usepackage{hypernat}
\usepackage[utf8]{inputenc}
\usepackage[T1]{fontenc}
\usepackage{cleveref}
\usepackage{soul}

\declaretheorem[parent=section,name=Theorem]{thm}

\declaretheorem[style=definition,sibling=thm]{definition}

\declaretheorem[sibling=thm]{proposition}
\declaretheorem[sibling=thm]{corollary}

\begin{document}
\title{Relating compatibility and divisibility of quantum channels}

\author{Cristhiano Duarte}
\affiliation{Wigner Research Centre for Physics, H-1121, Budapest Hungary}
\affiliation{Institute for Quantum Studies, Chapman University, One University Dr., Orange, CA 92866, USA}
\affiliation{International Institute of Physics, Federal University of Rio Grande do Norte, 59070-405 Natal, Brazil}
\email[Corresponding author: ]{crsilva@chapman.edu}
\author{Lorenzo Catani}
\affiliation{Schmid College of Science and Technology, Chapman University, One University Drive, Orange, CA, 92866, USA}
\author{Raphael C. Drumond}
\affiliation{Departamento de Matem\'{a}tica, Instituto de Ci\^{e}ncias Exatas, Universidade Federal de Minas Gerais, 30123-970, Belo Horizonte, Minas Gerais, Brazil}

\date{\today}
\begin{abstract}


We connect two key concepts in quantum information: compatibility and divisibility of quantum channels. Two channels are \textit{compatible} if they can be both obtained via marginalization from a third channel. A channel \textit{divides} another channel if it reproduces its action by sequential composition with a third channel. (In)compatibility is of central importance for studying the difference between classical and quantum dynamics. The relevance of divisibility stands in its close relationship with the onset of Markovianity.
We emphasize the simulability character of compatibility and divisibility, and, despite their structural difference, we find a set of channels -- self-degradable channels -- for which the two notions coincide. We also show that, for degradable channels, compatibility implies divisibility, and that, for anti-degradable channels, divisibility implies compatibility. These results motivate further research on these classes of channels and shed new light on the meaning of these two largely studied notions.

\end{abstract}

\maketitle

\section{Introduction}\label{Sec.Introduction}

Quantum theory is so far the best available framework to describe the microscopic world. It consists of a normative set of assumptions and rules designed to deal with phenomena with no classical explanation. Its predictive power is undeniable. Nonetheless, it also opens up room for puzzling behaviours having no classical counterpart~\cite{Bell66, Bell64, Kochen67, WZ82, JL15}. The root cause of all these remarkably odd phenomena is still unclear~\cite{Spekkens07}, and even quantum-to-classical and classical-to-quantum transitions are yet to be fully understood -- although, from a dynamical perspective, there have been robust proposals for obtaining such understanding~\cite{Zurek09,BZ06,BPH15}. 

What is undebatable, though, is that incompatibility of observables is a key feature of quantum theory ~\cite{WPGF09}. In this work we focus on the generalization of this notion to quantum channels. Intuitively, \textit{compatibility} between two quantum channels can be seen as on-demand simulability of those channels via a third, larger CPTP map~\cite{HMZ16, HM17, HKRS15, GPS20, CHMT19, Mori20}. In other words, this third channel contains, at all times, the information about the two original maps. Although we do not rule out the natural parallel with joint measurability~\cite{GBTCA17, WPGF09, BAN20},  we do not define the compatibility of two compatible channels as the fact that they can be realized simultaneously. Rather than saying that two compatible channels can be jointly realizable, we define compatible quantum channels as those channels that can be recovered from a third quantum map via marginalization (figure \ref{Fig.DefCompatibility}) ~\cite{GPS20}. It is within this perspective that we associate compatibility and simulability. 
\begin{figure}
\includegraphics[scale=0.25]{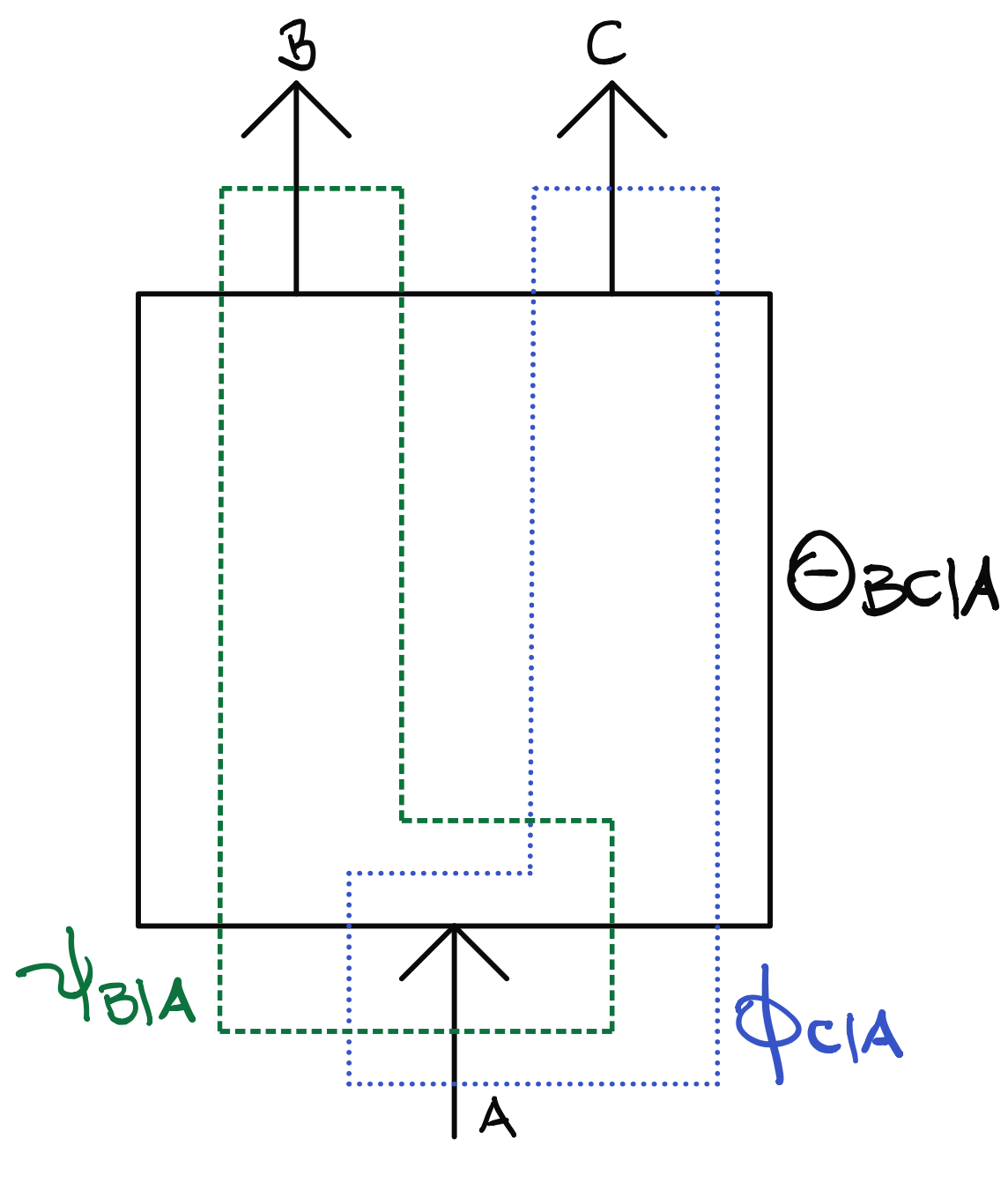}
\caption{Schematic box-representation of compatible maps $\psi_{B|A}$ and $\phi_{C|A}$. Their compatibilizer is the larger black box representing $\theta_{BC|A}$ (colours online).\label{Fig.DefCompatibility}}
\end{figure}

Our second main concept, \textit{divisibility}, has a long history within the open quantum dynamics community~\cite{GKS76,Lindblad79,KW05,MPM17}, where it is traditionally equaled to memoryless, or Markovian, processes~\cite{MKPFM19, BD16, RHP14}. Classical Markovian processes are governed by the Chapman-Kolmogorov equations and it turns out that a \textit{divisible dynamics} is described by a functional expression that satisfies a similar set of equations ~\cite{MSPM20}. More recently, it has become clearer that memoryless quantum processes deserve a different and less involved treatment~\cite{PollockEtAl18,CTZ08} and that divisibility makes no explicit reference to memory. We say that a quantum channel $\Psi$ divides another channel $\Phi$ when it is possible to find a third CPTP map $\theta$ such that $\Phi=\theta \circ \Psi$, where $\circ$ denotes sequential composition (figure \ref{Fig.DefDivisibility}).  In this sense, it is clear that divisibility of maps should not be seen  as characterizing memory, but rather as a signal of how one can simulate the action of a given map with the aid of two others.  
\begin{figure}
\includegraphics[scale=0.17]{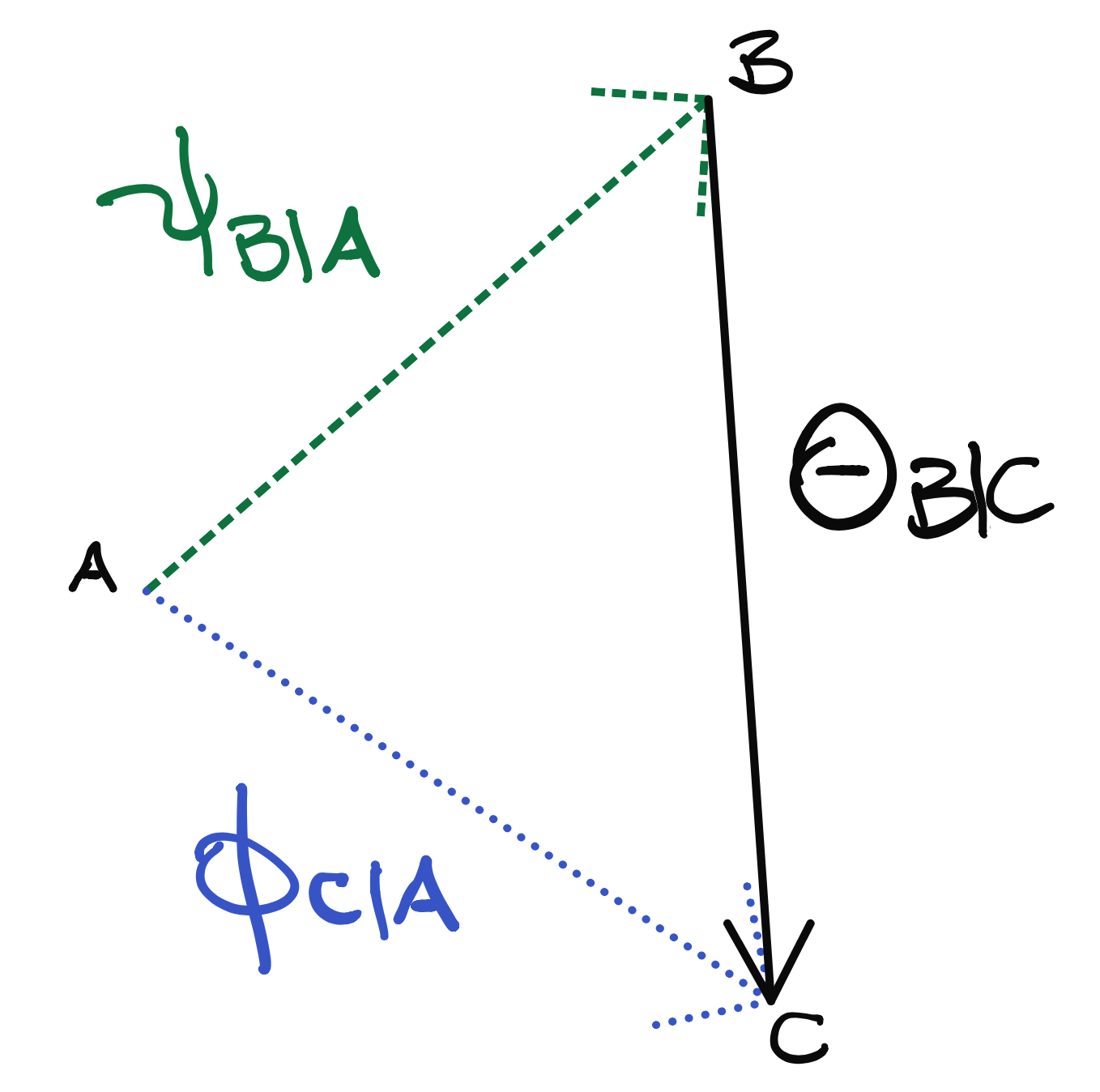}
\caption{Schematic arrow-representation  for $\psi_{B|A}$ dividing $\phi_{C|A}$. The quotient map $\theta_{C|B}$ is represented by the black arrow from $\mathcal{H}_{B}$ to $\mathcal{H}_{C}$ and verifies $\phi_{C|A}=\theta_{C|B} \circ \psi_{B|A}$ (colours online.)\label{Fig.DefDivisibility}}
\end{figure}

So far we have emphasized the simulability aspects exhibited by compatibility and divisibility. Undeniably, there are many other intriguing features associated with both concepts~\cite{WPGF09,DAMP05,BarnumEtAl96} , however, in this work, it is exactly the idea that they allow for some sort of simulability that we want to put forward and use as a motivation to bring them together. Compatibility amounts for the existence of a third quantum channel holding all the necessary information to simulate the original channels, whereas divisibility says that one of the maps can be simulated as a post-processing of the other. In both it is clear the existence of an external channel whose use can reveal the action of (at least one of) the original channels. Simply put, the former can be seen as a sort of simulability in parallel, and the latter as expressing a sequential-like simulability.  

Even in the light of this shared simulation feature, at a first glance, one may consider compatibility and divisibility as radically different concepts, with radically distinct body of applications. Our contribution comes to show that this is not entirely the case. Depending on the physical properties of the involved channels and on how the environment exchanges information with the evolving system, it might be the case that compatibility \emph{implies} divisibility, or that divisibility \emph{implies} compatibility, or  that compatibility  \emph{equals} divisibility. To obtain a full equivalence, we will have to touch on a third and unexpected ingredient: self-degradable channels~\cite{SSW08,CRS08,BDHM10, Holevo2007}. These are channels that allow one to obtain information about the evolution of the system looking only at the evolution of the environment and vice versa (see sec.~\ref{SubSec.ModifiedConjectures} for the precise definition).

This paper is organized as follows. In sec.~\ref{Sec.Definitions} we define compatibility and divisibility. Sec.~\ref{Sec.Conjectures} makes explicit that a connection for generic channels is not possible, as neither compatibility directly implies divisibility nor divisibility implies compatibility; we exemplify each case and discuss the reasons of this impossibility. It is in sec.~\ref{Sec.ConjecturesProofs} that we start establishing the desired bridge between the two concepts. To do so, firstly we define a very important property of quantum channels -- degradability. Secondly, in subsec.~\ref{SubSec.ModifiedConjectures}, we modify our initial attempts, bring in the idea of (self-)degradable channels, and establish our main result. In subsec.~\ref{SubSec.DiscussionCorrect} we provide a physical explanation for why degradability is the key concept that makes the connection possible. Finally, sec.~\ref{Sec.Conclusion} is where we conclude our work. The structure of the manuscript is chosen with the goal to help the reader to understand why the connection between compatibility and divisibility is so surprising.

\section{Main Definitions} \label{Sec.Definitions}

To facilitate the reading, we will stick with the notation of refs.~\cite{LS13, Leifer2014}. In other words, 
\begin{equation}
\psi_{Y|X}: \mathcal{D}(\mathcal{H}_{X}) \rightarrow \mathcal{D}(\mathcal{H}_{Y}) 
\end{equation}
 will always represents a quantum map from $\mathcal{D}(\mathcal{H}_{X})$ to $\mathcal{D}(\mathcal{H}_{Y})$. The sub-index $Y|X$ is in parallel with the standard conditional probability language and ought to be read as \textit{Y given X}.

The Jamio\l kowski isomorphic image of each quantum channel $\psi_{Y|X}$ is a non-normalized state in $\mathcal{D}(\mathcal{H}_{X}) \otimes \mathcal{D}(\mathcal{H}_{Y})$ defined via:
\begin{equation}
J_{\psi_{Y|X}}:=\sum_{i,j=1}^{\mbox{dim}(\mathcal{H}_{X})}\ket{i}\bra{j} \otimes \psi_{Y|X}(\ket{i}\bra{j}).
\label{Eq.DefJamiolkowski}
\end{equation}
We refer to~\cite{Watrous2018} for an introductory discussion about the (Choi-)Jamio\l kowiski operator. Further properties of the isomorphism can be found in~\cite{Gour19,LS13}.

The sections below contain the standard definitions of divisibility and compatibility. These two topics are sub-branches of quantum information and lines of research in their own, and, as such, we will not treat them comprehensively here. To help the reader, we will briefly motivate the main definitions and point to the standard literature.

\subsection{Divisibility}\label{SubSec.Divisibility}

 It has become broadly accepted that discrete evolution of open quantum systems should be represented by a family $\mathcal{F}=\{\psi_{X_{0}|X_{k}}\}_{k=1}^{T}$ of CPTP maps~\cite{BD16, NC10,MPM17}. In this family, each channel 
\begin{equation}
 \psi_{X_{0}|X_{k}}: \mathcal{D}(\mathcal{H}_{X_{0}}) \rightarrow \mathcal{D}(\mathcal{H}_{X_{k}})
\end{equation}
represents a quantum process happening from a initial time step $t=t_0$ to a posterior time $t=t_k$. Given a quantum system evolving according to the family $\mathcal{F}$, if $\rho_{0}$ represents the state of this system at $t=t_0$, its evolved state at $t=t_k$ is given by $\rho_{k}=\psi_{X_{0}|X_{k}}(\rho_{0})$.

This definition captures the idea that there is some quantum process happening between two time steps. However, the $t_0$-to-$t_k$ evolution is not natural in a dynamical sense. Instead of having a family dictating the evolution from $t_0$ to $t_k$, one would expect to have a sequential family of CPTP maps, say $\mathcal{F}^{\prime}=\{\psi_{X_{k+1}|X_{k}}\}_{k=0}^{T-1}$, governing how the system changes between to consecutive time steps $t_k$-to-$t_{k+1}$. It is the existence of this second collection of CPTP maps that divisibility, or Markovianity, talks about~\cite{BD16, RHP14, MKPFM19}.

\begin{definition} [CP-Divisibility]
Let $\mathcal{F}=\{\psi_{X_{0}|X_{k}}\}_{k=1}^{T}$ be a family of  maps representing the dynamics of an open quantum system. Each map
\begin{equation}
 \psi_{X_{0}|X_{k}}: \mathcal{D}(\mathcal{H}_{X_{0}}) \rightarrow \mathcal{D}(\mathcal{H}_{X_{k}})
\end{equation}
being a completely positive trace preserving map. We say that this family is divisible, or that the dynamics is Markovian,  whenever there exists another family of maps $\mathcal{F}^{\prime}=\{\psi_{X_{k+1}|X_{k}}\}_{k=0}^{T-1}$ such that
\begin{enumerate}
\item $\forall \,\, k:$ $\psi_{X_{k+1}|X_{k}}: \mathcal{D}(\mathcal{H}_{X_{k}}) \rightarrow \mathcal{D}(\mathcal{H}_{X_{k+1}})$ is CPTP
\item $\forall \,\, k:$ $\psi_{X_{k+1}|X_{0}}=\psi_{X_{k+1}|X_{k}} \circ \psi_{X_{k}|X_{0}}$.
\end{enumerate}
\label{Def.Divisibility}
\end{definition}
 
In other words, def.~\ref{Def.Divisibility} says that a family of $t_0$-to-$t_k$ quantum channels is divisible when it is possible to find intermediate channels $t_{k}$-to-$t_{k+1}$ that ``match'' the original family. This matching is expressed by stating that it is possible to simulate any process happening from $t_{0}$-to-$t_{k+1}$ provided that it is known what happens in $t_0$-to-$t_k$ and $t_k$-to-$t_{k+1}$.  Fig.~\ref{Fig.CPTP_Maps_Divisibility} illustrates this argument.

\begin{figure}
    \includegraphics[scale=0.25]{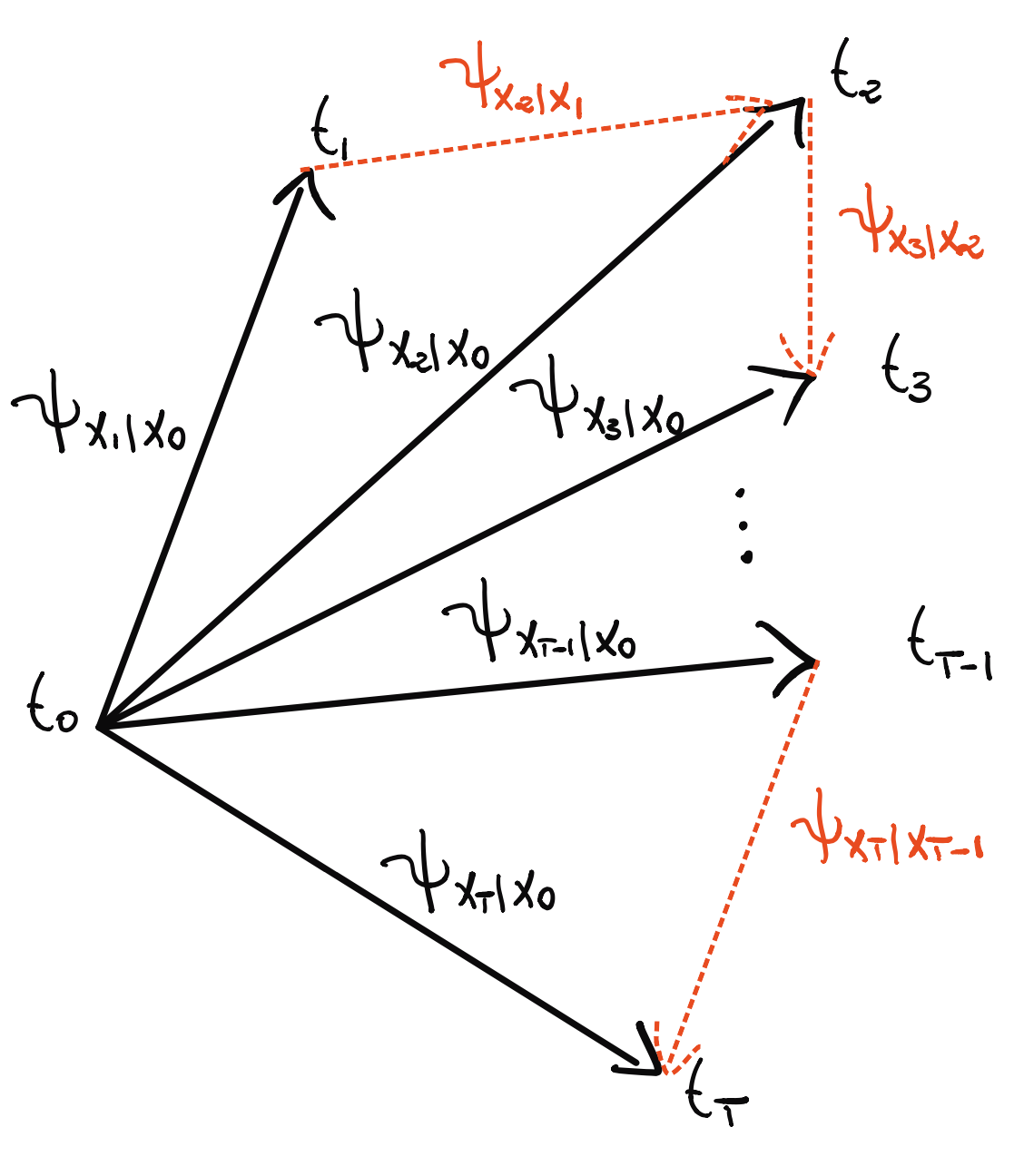}
    \caption{Arrow-like representation of a Markovian quantum evolution. Black solid arrow mean $t_0$-to-$t_k$ quantum processes. Orange dashed arrows are the intermediate $t_k$-to-$t_{k+1}$ quantum processes (colours online). \label{Fig.CPTP_Maps_Divisibility}}
\end{figure}

\textbf{Remark:} In def.~\ref{Def.Divisibility} we used the term Markovianity also to express memoryless processes. This is still a topic of ongoing research and intense debate~\cite{MKPFM19, PollockEtAl18, SRS19}. Although we are more inclined to agree with ref.~\cite{MKPFM19}, which expresses more cautiousness in associating divisibility and memoryless processes, the lack of memory of divisible processes is often formulated drawing a parallel with the classical case. Basically, a classical stochastic process is a family $\{X_k\}_{k}$ of random variables where it is possible to define a joint $\sigma-$algebra. This process is Markovian when $\mathbb{P}(X_k|X_{k-1},...,X_1) = \mathbb{P}(X_k|X_{k-1})$ holds true for every $k$. In plain English, the information necessary do determine what will happen at $t_k$ is fully determined by the knowledge of what happened at $t_{k-1}$, hence a memoryless process. The fact is that Markovian processes have a neat equation governing $\mathbb{P}(X_k|X_{k-1})$ that only involves intermediate time-steps: the Chapman-Kolmogorov equation~. Roughly speaking, if $\{X_k\}_{k}$ is Markovian, then for every $t>k>r$~\cite{MSPM20}:
\begin{equation}
\mathbb{P}(X_t|X_r)=\sum_{x_k}\mathbb{P}(X_t|X_{k}=x_k)\mathbb{P}(X_k=x_k|X_{r}).
\label{Eq.ChapmanKolomogorov}
\end{equation}
It is exactly the parallel between eq.~\eqref{Eq.ChapmanKolomogorov} and def.~\ref{Def.Divisibility} that justifies the comparison between the classical and quantum cases, and motivates the use of the name ``Markovianity'' for the quantum case. We refer to~\cite{MKPFM19} for a thorough analysis of this topic.

Specializing to the case of only two maps, we can write down a more direct definition of divisibility. As a matter of fact, this variant definition shows not only an asymmetry which is always hidden away in the usual definition of divisibility, but it also emphasizes another mathematical aspect present in the original definition. 

\begin{definition}
Let $\psi_{B|A}$ and $\phi_{C|A}$ be two cptp maps. We say that $\psi_{B|A}$ \textit{divides} $\phi_{C|A}$ whenever there exists another cptp map $\theta_{C|B}$ such that 
\begin{equation}
\phi_{C|A}=\theta_{C|B} \circ \psi_{B|A}.
\label{Eq.DefDivisibilityTwoMaps}
\end{equation}
We will refer to $\theta_{C|B}$ as the quotient map.
\label{Def.DivisibilityTwoMaps}
\end{definition}

When considering only two maps, the alternative def.~\ref{Def.DivisibilityTwoMaps} makes clearer the simulability rationale: to obtain the action of $\phi_{C|A}$ it suffices to know $\theta_{C|B}$ and $\psi_{B|A}$. Broadly speaking, we can interpret the map $\theta_{C|B}$ as containing the piece of information that is necessary to define $\phi_{C|A}$ via composition with $\psi_{B|A}$. We will adopt a similar perspective when interpreting compatibility in the next section.

\subsection{Compatibility}\label{SubSec.Compatibility}

One of the foundational cornerstones of quantum theory is the existence of observables that cannot be simultaneously, or jointly, measured. In other words, quantum theory does not rule out the presence of incompatible observables. This fact is very well known, and goes all the way back to the Heisenberg's uncertainty principle. Moreover, we can say that incompatibility is crucial to the foundations of quantum physics, as incompatible measurements are the root cause allowing for violations of Bell inequalities~\cite{WPGF09}. As a matter of fact, quantum (in)compatibility is not restricted only to POVM's. This short section shows how to lift the concept of incompatibility from measurements to quantum channels. 

Broadly speaking, two quantum channels are compatible when they can be enlarged, seen as being just one, and have their action recovered by tracing out part of this larger structure. The idea behind the concept of compatibility, therefore, is associated to the existence of a larger quantum channel that can, on-demand, be used to recover the action of the two original quantum channels. Precisely put~\cite{HMZ16},

\begin{definition}[Compatibility]
Let $\psi_{B|A}$ and $\phi_{C|A}$ be two cptp maps. We say that $\psi_{B|A}$ \textit{is compatible} with $\phi_{C|A}$ whenever there exists another cptp map 
\begin{equation}
\theta_{BC|A}: \mathcal{D}(\mathcal{H}_{A}) \rightarrow \mathcal{D}(\mathcal{H}_{B}) \otimes \mathcal{D}(\mathcal{H}_{C})
\end{equation}
such that, for every $\rho_{A}$ in $\mathcal{D}(\mathcal{H}_{A})$, the following equations
\begin{equation}
 \psi_{B|A}(\rho_{A})=\mbox{Tr}_{C}\left[\theta_{BC|A}(\rho_{A}) \right]
 \label{Eq.DefCompatibility1}
\end{equation}
\begin{equation}
 \phi_{C|A}(\rho_{A})=\mbox{Tr}_{B}\left[\theta_{BC|A}(\rho_{A}) \right]
 \label{Eq.DefCompatibility2}
\end{equation}
hold true. The channel $\theta_{BC|A}$ is called the compatibilizer for $\psi_{B|A}$ and $\phi_{C|A}$. When it is not possible to find a channel satisfying eqs.~\eqref{Eq.DefCompatibility1} and~\eqref{Eq.DefCompatibility2} above, we say that $\psi_{B|A}$ and $\phi_{C|A}$ are \textit{incompatible}.
\label{Def.Compatibility}
\end{definition} 

\textbf{Example 1: }The most paradigmatic example of incompatibility, that will be used later on, is related to the no-broadcasting theorem~\cite{HM17, BarnumEtAl96, GPS20}: consider $\psi_{A|A}=\mbox{id}_{A}=\phi_{A|A}$, where $\mbox(id)_{A}:\mathcal{D}(\mathcal{H}_{A}) \rightarrow \mathcal{D}(\mathcal{H}_{A})$ is the identity channel. If there existed a compatibilizer for this case, there would exist a linear map $\theta_{AA|A}$ such that:
\begin{equation}
\forall \, \rho_{A} \in \mathcal{D}(\mathcal{H}_{A}): \rho_{A} = \mbox{Tr}_{A}\left[\theta_{AA|A}(\rho_{A})\right],
\end{equation}
which would, in turn, imply a universal broadcast machine. We are forced to conclude that the identity map id$_{A}$ is not self-compatible~\cite{HM17}. This example shows also that compatibility of quantum channels is mathematically rich and complex, despite having a simple definition and a very straightforward  information theoretical interpretation.

\textbf{Example 2:} The second example shows how we can construct as many compatible channels as we want and, on top of that, the following construction also explicitly exhibits the compatibilizer. Given two non-trivial and finite-dimensional Hilbert spaces $\mathcal{H}_{B}$ and $\mathcal{H}_{C}$, define $\mathcal{H}_{A}:=\mathcal{H}_{B} \otimes \mathcal{H}_{C}$. Now, chose $\tilde{\psi}_{B|B}$ and $\tilde{\phi}_{C|C}$ arbitrarily. Finally, define ${\psi}_{B|A}$ and ${\psi}_{C|A}$ acting on $\mathcal{D}(\mathcal{H}_{A})$ via:
\begin{equation}
\psi_{B|A}(\rho_{A}):=\left(\tilde{\psi}_{B|B} \circ \mbox{Tr}_{C} \right) \left( \rho_{A} \right)
\label{Eq.ExampleFernandinho1}
\end{equation}
\begin{equation}
\phi_{C|A}(\rho_{A}):=\left(\tilde{\phi}_{C|C} \circ \mbox{Tr}_{B} \right) \left( \rho_{A} \right),
\label{Eq.ExampleFernandinho2}
\end{equation}
for every $\rho_A\in\mathcal{H}_{A}$.
Now, the main claim is that ${\psi}_{B|A}$ and ${\psi}_{C|A}$ are compatible and, additionally, that $\tilde{\psi}_{B|B} \otimes \tilde{\phi}_{C|C}$ is their compatibilizer. As a matter of fact:
\begin{align}
      \forall \,\, &\rho_{A} \in \mathcal{D}(\mathcal{H}_{A}): \mbox{Tr}_{C}\left[ (\tilde{\psi}_{B|B} \otimes \tilde{\phi}_{C|C}) \rho_{A} \right]  =    \nonumber \\
      & = \sum_{i,j}\sum_{k,l}\alpha_{ij}^{kl}  \mbox{Tr}_{C}\left[ (\tilde{\psi}_{B|B} \otimes \tilde{\phi}_{C|C}) \ket{i}\bra{j}_{B} \otimes \ket{k}\bra{l}_{C} \right] \nonumber \\
      &= \sum_{i,j}\sum_{k,l}\alpha_{ij}^{kl}  \mbox{Tr}_{C}\left[ \tilde{\phi}_{C|C} (\ket{k}\bra{l}_{C}) \right] \tilde{\psi}_{B|B}(\ket{i}\bra{j}_{B}) \\
      &= \sum_{i,j}\sum_{k}\alpha_{ij}^{kk} \tilde{\psi}_{B|B}\left(\ket{i}\bra{j}_{B}\right) =  \tilde{\psi}_{B|B}\left(\sum_{i,j}\sum_{k}\alpha_{ij}^{kk}\ket{i}\bra{j}_{B}\right) \nonumber \\
      &= \left(\tilde{\psi}_{B|B} \circ \mbox{Tr}_{C} \right) \left( \rho_{A} \right) = \psi_{B|A}(\rho_{A}). \nonumber
\end{align}
With the same argument it can be shown that $\mbox{Tr}_{C}\left[ (\tilde{\psi}_{B|B} \otimes \tilde{\phi}_{C|C}) \rho_{A} \right] = \phi_{C|A}(\rho_{A})$. In conclusion, not only the maps defined in eqs.~\eqref{Eq.ExampleFernandinho1} and~\eqref{Eq.ExampleFernandinho2} are compatible, but also the factorized map $\tilde{\psi}_{B|B} \otimes \tilde{\phi}_{C|C}$ is their compatibilizer. 

Intuitively, two ingredients are necessary to run the second example. First, we have defined $\mathcal{H}_{A}$ to be the tensor product of $\mathcal{H}_{B}$ and $\mathcal{H}_{C}$; we have opened room for parallel and independent action of the two channels. Secondly, we have explored this room and defined $\psi_{B|A}$ as if it was acting only on $B$ and, likewise, $\phi_{C|A}$ as if it was acting only on $C$. Even though there might be correlations between $B$ and $C$ codified in $\rho_{A}$, the knowledge of what is happening in $A$ renders $B$ and $C$ conditionally independent~\cite{LS13,LP08}. Basically, in this example, the compatibility of the two maps is a direct consequence of these two points. The tensor structure of $A$ opened up space for enlarging the maps, and their conditional independence was crucial to remove either one on-demand.  

\textbf{Example 3: } This last example shows that, in contrast with the resource theory of entanglement~\cite{PV07, HHH09}, there is no catalysis~\cite{JP99, Fritz2017} of compatibility . More precisely, we will see that if $\psi_{B|A} \otimes\chi_{B'|A'}$ is compatible with $\phi_{C|A} \otimes \chi_{B''|A'}$, then $\psi_{B|A}$ and $\phi_{C|A}$ are also compatible. As a matter of fact, compatibility between $\psi_{B|A} \otimes\chi_{B'|A'}$ and $\phi_{C|A} \otimes \chi_{B''|A'}$ implies the existence of $\tilde{\theta}_{BB'CB''|AA'}$ such that
\begin{equation}
\psi_{B|A} \otimes \chi_{B'|A}(\rho_{AA'}) = \mbox{Tr}_{CB''}\left[ \tilde{\theta}_{BB'CB''|AA'}(\rho_{AA'}) \right]
\label{Eq.ExampleNoCatalysis1}
\end{equation}
and
\begin{equation}
\phi_{C|A} \otimes \chi_{B''|A}(\rho_{AA'}) = \mbox{Tr}_{BB'}\left[ \tilde{\theta}_{BB'CB''|AA'}(\rho_{AA'}) \right]\label{Eq.ExampleNoCatalysis2}
\end{equation}
for every $\rho_{AA'}$. We define a compatibilizer for $\psi_{B|A}$ and $\phi_{C|A}$ using eqs.~\eqref{Eq.ExampleNoCatalysis1} and~\eqref{Eq.ExampleNoCatalysis2} above as:
\begin{equation}
\theta_{BC|A}(\rho_{A}):=\mbox{Tr}_{B'B''}\left[ \tilde{\theta}_{BB'CB''|AA'}\left(\rho_{A} \otimes \frac{\id_{A'}}{d_{A'}} \right) \right],
\label{Eq.ExampleNoCatalysisDefiningTheta}
\end{equation}
for each $\rho_{A}$, where d$_{A'}$ is the dimension of $\mathcal{H}_{A'}$. To check that $\theta_{BC|A}$ is a valid compatibilizer, it suffices to trace either system out of eq.~\eqref{Eq.ExampleNoCatalysisDefiningTheta}:
\begin{align}
\forall \,\, \rho_{A}& \in \mathcal{D}(\mathcal{H}_{A}): \mbox{Tr}_{C}\left[  \theta_{BC|A}(\rho_{A}) \right]= \nonumber \\
&=\mbox{Tr}_{CB'B''}\left[ \tilde{\theta}_{BB'CB''|AA'}\left(\rho_{A} \otimes \frac{\id_{A'}}{d_{A'}} \right) \right] \nonumber \\
&= \mbox{Tr}_{B'CB''}\left[ \tilde{\theta}_{BB'CB''|AA'}\left(\rho_{A} \otimes \frac{\id_{A'}}{d_{A'}} \right) \right]  \\
&=    \mbox{Tr}_{B'}\left[ \psi_{B|A} \otimes \chi_{B'|A'}\left( \rho_{A} \otimes \frac{\id_{A'}}{d_{A'}}  \right) \right]  \nonumber \\
&=    \psi_{B|A}(\rho_{A}) \mbox{Tr}\left[ \chi_{B'|A'}\left( \frac{\id_{A'}}{d_{A'}}  \right) \right] \nonumber \\
&=    \psi_{B|A}(\rho_{A}) \nonumber.
\end{align}
Similarly, the same argument leads to $\mbox{Tr}_{B}\left[  \theta_{BC|A}(\rho_{A}) \right]= \phi_{C|A}(\rho_{A})$ for every $\rho_{A}$. To sum up, our construction says that it does not matter if an auxiliary channel is plugged into $\psi_{B|A}$ and $\phi_{C|A}$. If those channels are not compatible to begin with, they will remain incompatible even with the help of an external channel $\chi_{B'|A}$. In other words, in the sense of~\cite{JP99, Fritz2017} there is no catalysis of compatibility.
 
As will become clearer, the examples we explore here have only been chosen because they will be useful to investigate the connection between compatibility and divisibility later on. Nonetheless, we hope that they have been sufficiently illustrative to the reader who is unfamiliar with the field. Compatibility constitutes an active research topic, and this short sub-section does not represent an exhaustive review of it.  For recent developments on compatibility we suggest~\cite{HM17} and~\cite{GPS20}.

\section{Compatibility and divisibility are not connected in general}\label{Sec.Conjectures}

As discussed in sec.~\ref{Sec.Introduction}, divisibility and compatibility can be both associated with the notion of simulability. 
Although the simulability encoded in def.~\ref{Def.Compatibility} comes, intuitively, in a parallel form and the simulability of def.~\ref{Def.DivisibilityTwoMaps} in a sequential form, they both consist of writing down a map in terms of others. This is exactly the original motivation behind our initial conjectures. 

As we will see in sub-sec.~\ref{SubSec.OriginalConjectures}, these conjectures do not hold true in general, and sub-sec.~\ref{SubSec.DiscussionWrong} addresses what is wrong with them. The main reason lies indeed in the conflict between the sequential and parallel characters of the two notions. 

Nonetheless, it is remarkable that if we restrict our argument to a certain class of important quantum channels, we can connect the two concepts. Sec.~\ref{Sec.ConjecturesProofs} is addressing this point. For now, we think it is instructive to discuss our original guesses and the counter-examples to them.   

\subsection{Original Conjectures and Counter-Examples}\label{SubSec.OriginalConjectures}

Let us start with our first attempt. Conjecture 1 states that divisibility should imply compatibility. In other words, it says that parallel simulability could be turned into sequential simulability.

\textbf{Conjecture 1: } If $\psi_{B|A}$ divides $\phi_{C|A}$, then it must be the case that $\psi_{B|A}$ is compatible with $\phi_{C|A}$.

It is not so difficult to come up with a counter-example for this statement. As it turns out, we have already discussed a counter-example for this conjecture before, more precisely in sec.~\ref{SubSec.Compatibility}. Referring back to example 1, it suffices to take $\mathcal{H}_{A}=\mathcal{H}_{B}=\mathcal{H}_{C}$ and $\psi_{B|A}=\mbox{id}_{A}=\phi_{C|A}$. Even though $\psi_{B|A}$ divides $\phi_{C|A}$ -- with the quotient map being nothing but $\theta_{C|B}=\mbox{id}_{A}$, there is no compatibilizer for this case. In conclusion, our first conjecture does not hold true in general. 

Because we are building our argument upon the non-universal broadcasting theorem, we can formulate, in this context, a follow-up question: what if one requires more structure on $A$? For instance, if we restricted ourselves, and considered divisibility and compatibility only on $\mathcal{A}$ , a commuting sub-algebra of $\mathcal{D}(\mathcal{H}_{A})$, then in this case our conjecture would follow~\cite{CHMT19}. As we will be focusing on general properties of channels, in this work we will not approach restricted domains for the channels. We leave this point open for further investigations.

Our second conjecture regards the converse direction. It explores the possibility of having sequential-like simulability as a building block for parallel simulability.

\textbf{Conjecture 2: } If $\psi_{B|A}$ is compatible with $\phi_{C|A}$, then it must be the case that $\psi_{B|A}$ divides $\phi_{C|A}$.

As we had already anticipated, this conjecture also does not hold true. To show this we will use again a previous example. We have learned from Example 2, in sec.~\ref{Def.Compatibility}, how to construct compatible maps $\psi_{B|A}$ and $\phi_{C|A}$ whose compatibilizer is $\tilde{\psi}_{B|B}$ and $\tilde{\phi}_{C|C}$. It suffices, for instance, to consider the case where $\tilde{\psi}_{B|B}$ is the completely depolarizing map~\cite{NC10} and $\tilde{\phi}_{C|C} = \mbox{id}_{C}$. If there existed a quotient map $\theta_{C|B}$ for this example, it would be a many-to-one relation, and therefore not even a function.  

In conclusion, both conjectures do not hold true at this level of generality, and we will try and explain the reasons why in the section below.

\subsection{Discussion: what went wrong?}\label{SubSec.DiscussionWrong}

 To begin with, conjecture 1 expresses the idea that divisible channels can be seen as a realization of a larger channel that allows to recover the original channels by the appropriate tracing out. 
This is indeed a too strong claim. First, note that if we assume divisibility, the quotient channel has, in general, a preferred direction; $\psi_{B|A}$ dividing $\phi_{C|A}$ implies a channel from $B$ to $C$ such that $\phi_{C|A}=\theta_{C|B} \circ \psi_{B|A}$. In this sense, not only the simulability is a sequential-like simulability, but also it has a preferred direction. In opposition, because of the parallel character, compatibility is more symmetrical. This very aspect also plagues conjecture 2, as we will see in the next paragraph. Secondly, as we pointed out before, the set of channels considered in our conjecture is too broad and too general. We mentioned that if one is willing to pay the cost of restricting the domain of the involved channels to something smaller and with more structure, it is the case that the trade-off must secure a truth value for the statement. Although we think this is a point deserving further investigation, we will not do so here, as we are more concerned with which properties of the channels that can make the conjectures to be true. 

Conjecture 2 not only suffers from the same symmetry v. asymmetry issue already discussed before, but also from the tension between divisibility and the parallel character defining compatibility. The symmetry v. asymmetry tension manifested here is due to the lack of preferred direction of compatible channels. With no further assumptions, if compatibility implied $\psi$ dividing $\phi$, it should also be true that $\phi$ divided $\psi$, and we have already discussed a case ruling out this possibility. The second point, that of exploring the parallel simulability of the channels, is what we tried and explored in example 2. It shows that independent, or conditionally independent, channels might be compatible as well as not-divisible for the very same reason: their action is, broadly speaking, in parallel. Once again, although we think that a connection between conditional independent channels and compatibility is expected, we will not investigate any further their relationship in this work. 

We have laid down some reasons why our initial conjectures are false. Plus, we have also mentioned two possible avenues for further investigations. Mathematically, we are missing a piece that would make the connection we want to establish possible. We find it in the next section. 

\section{Main Results and Proofs}\label{Sec.ConjecturesProofs}

We have established that neither divisibility implies compatibility nor vice versa. In this section we show that if we consider a smaller and very well-studied family of quantum channels it is possible to salvage the argument and remarkably connect divisibility with compatibility.

The concept we are missing is that of degradable channels~\cite{SSW08,CRS08,BDHM10, Holevo2007}.  We will define and motivate this family of quantum maps in sub-sec.~\ref{SubSec.DefDegradability}. After that, in sub-sec.~\ref{SubSec.ModifiedConjectures}, we modify and rewrite our initial conjectures in a form that turns out to be true. Finally, in sub-sec~\ref{SubSec.DiscussionCorrect},  we analyze the reasons for the success in connecting compatibility and divisibility when restricted to degradable channels.

\subsection{Complementary and Compatibility}\label{SubSec.DefDegradability}

We here follow the standard treatment for describing the evolution of an open quantum system: for not being isolated, it interacts with a (multipartite) environment; initially the pair system-environment is described by a product state; following a closed dynamics, system and environment evolve according to a family of unitary maps $\mathcal{U}_{t}$; the state of the system is, at any posterior time, obtained by tracing out the degrees of freedom of the environment~\cite{Watrous2018,BPH15}. The resultant of this process is a family of cptp maps $\{ \psi_{t|0} \}_{t}$ where the environment is always being traced-out:
\begin{equation}
\psi_{t|0} (\rho_{S})=\mbox{Tr}_{E}\left[U^{\ast} (\rho_{S} \otimes \ket{0}\bra{0}) U\right].
\label{Eq.TracingOutEnv}
\end{equation}

What if, instead of tracing-out the environment, we trace out the system? The authors of~\cite{BPH15} have thoroughly explored the consequences of doing this. Here, closely following~\cite{SSW08,CRS08,BDHM10, Holevo2007} we will implement a more modest approach. Referring back to the main topic of this work, the main idea here is to show how one can use this natural way of constructing channels to reconcile divisibility and compatibility.  We begin with a definition~\cite{CRS08}:

\begin{definition}[Complementary Channel]
Let $\psi_{B|A}$ be a quantum channel and 
\begin{equation}
\psi_{B|A}(\cdot)= \mbox{Tr}_{E}\left[ V (\cdot) V^{\ast} \right]
\label{Eq.DefComplementaryChannelStinespring}
\end{equation}
be its Lindblad-Stinespring representation, where $V:\mathcal{H}_{A} \rightarrow \mathcal{H}_{B} \otimes \mathcal{H}_{E}$ is an isometry. The \emph{complementary channel} of $\psi_{B|A}$ is defined as:
\begin{equation}
\psi_{E|A}^{c}(\cdot)= \mbox{Tr}_{B}\left[ V (\cdot) V^{\ast} \right],
\label{Eq.DefComplementaryChannel}
\end{equation}
where the trace now is performed over $\mathcal{H}_{B}$.
\label{Def.ComplementaryChannel}
\end{definition}

Broadly speaking, the complementary channel $\psi^{c}_{E|A}$ describes the evolution of the environment in correspondence to the evolution of the system described by $\psi_{B|A}$.

\textbf{Remark:} There are alternative, equivalent ways to introduce complementary channels. For example, instead of starting from the Lindblad-Stinespring representation, if we had started from a Kraus decomposition $\psi_{B|A}(\cdot)=\sum_{i=1}^{M}K_{i}(\cdot)K_{i}^{\ast}$, the complementary map $\psi_{B|A}^{c}$ to this channel, with this particular Kraus representation, would be:
\begin{equation}
\psi_{E|A}^{c}(\cdot)=\sum_{i,j=1}^{M}\mbox{Tr}\left[ K_{i}^{\ast}K_{j}(\cdot) \right] \ket{j}\bra{i}_{E}.
\label{Eq.RemarkAlternativeDefComplementary}
\end{equation}
The connection between eqs.~\eqref{Eq.DefComplementaryChannel} and~\eqref{Eq.RemarkAlternativeDefComplementary} is established noticing that from an isometry $V:\mathcal{H}_{A} \rightarrow \mathcal{H}_{B} \otimes \mathcal{H}_{E}$ it is possible to define a valid set of Kraus operators $\{K_i:=V \otimes \ket{i}\}_{i=1}^{d_E}$ and, vice versa, that from a set of Kraus operators $\{K_{i}\}_{i=1}^{M}$ it is possible to form an isometry $V:=\sum_{i=1}^{M}K_i \otimes \ket{i}$. 

 The study of complementary channels goes back to the works of A. Holevo~\cite{Holevo2007}, M. B. Ruskai plus collaborators~\cite{CRS08,KMNR05} and I. Devetak with P. Shor~\cite{DS05}. This class of channels play a role in quantum information theory and, remarkably, as we mentioned before, in the foundations of quantum mechanics in the attempt to explain the emergence of an objective macroscopic reality~\cite{BPH15}. Expanding the former role a bit further, in this work, we will use the concept of complementarity as a mechanism to connect compatible with divisible maps. The theorem below (proven in ref.~\cite{HM17}) is the very first to do so.

\begin{thm}
Two quantum channels $\phi_{C|A}$ and $\psi_{B|A}$ are compatible if and only if there exists  a cptp map $\theta_{C|E}$ such that $\phi_{C|A}=\theta_{C|E} \circ \psi_{E|A}^{c}$.
\label{Thm.CompatibilityAndOrdering}
\end{thm}

In other words, the compatibility between $\phi_{C|A}$ and $\psi_{B|A}$ means that $\psi_{E|A}^{c}$ divides $\phi_{C|A}$. This theorem connects the parallel character of compatibility with the sequential character of divisibility (of the complementary channel).
It is also important to notice about thm.~\ref{Thm.CompatibilityAndOrdering} that it restores the symmetry lost in our first attempt. With the cost of considering the complement of a map not the map itself, the theorem says that not only we can write $\phi_{C|A}$ using a manipulated version of $\psi_{B|A}$, but also that the other way round does hold true: we can write $\psi_{B|A}$ via a modification of $\phi_{C|A}$. 

Secondly, within the proof of the thm~\ref{Thm.CompatibilityAndOrdering} it is already present the construction of a possible map making the passage between $\psi_{B|A}^{c}$ and $\phi_{C|A}$. We here provide a sketch of the proof. We refer to ~\cite{HM17} for an alternative and full proof of thm.~\ref{Thm.CompatibilityAndOrdering}. 

Assuming $\psi_{B|A}$ and $\phi_{C|A}$ compatible, their compatibilizer $\theta_{BC|A}$ has a Lindblad-Stinespring form
\begin{equation}
\theta_{BC|A}(\rho_{A})=\mbox{Tr}_{E}\left(V\rho_{a}V^{\ast}\right),
\label{Eq.LSFormCompatiblizer}
\end{equation}
where $V:\mathcal{H}_{A} \rightarrow (\mathcal{H}_{B} \otimes \mathcal{H}_{C}) \otimes \mathcal{H}_{E}$ is an isometry. Now, we can extract an alternative Lindblad-Stinespring form for $\psi_{B|A}$ via eq.~\eqref{Eq.LSFormCompatiblizer}, and use this alternative form to define a complementary map. This direct construction will suffice to conclude the argument. As a matter of fact,
\begin{equation}
\psi_{B|A}(\rho_{A}) = \mbox{Tr}_{C}\left[ \theta_{BC|A}(\rho_{A}) \right] = \mbox{Tr}_{CE}\left[ V \rho_{A} V^{\ast}\right],
\end{equation}
and considering $\mathcal{H}_{C} \otimes \mathcal{H}_{E}$ as a new environment, we can define 
 \begin{equation}
  \psi_{CE|A}^{c}:=\mbox{Tr}_{B}\left[ V \rho_{A} V^{\ast}\right].
 \end{equation}
To conclude,
\begin{align}
\phi_{C|A}(\rho_{A}) &= \mbox{Tr}_{B}\left[ \theta_{BC|A}(\rho_{A}) \right] = \mbox{Tr}_{BE}\left[ V \rho_{A} V^{\ast}\right] \nonumber \\
&= \mbox{Tr}_{E} \circ \mbox{Tr}_{B}\left[ V \rho_{A} V^{\ast}\right] = \mbox{Tr}_{E} \circ \psi_{CE|A}^{c} \nonumber \\
&= \theta_{C|E} \circ \psi_{CE|A}^{c},
\end{align}
by simply rewriting $\theta_{C|E}$ as the partial trace. 

Third, in thm.~\ref{Thm.CompatibilityAndOrdering} we are tacitly assuming the existence of one particular Lindblad-Stinespring dilation for $\psi_{B|A}$. Even though the result holds true for a particular dilation, and its equivalence class, it might be the case that for other distinct dilations the result does not hold true~\cite{HM17}. Although this is an important point to bear in mind, it does not affect nor weaken our main results. It suffices to think of $\psi_{B|A}$ as the quantum channel originating from the interaction with a fixed environment $E$, so its action and, more importantly, its dilation is already naturally given by $\psi_{B|A}(\rho_{A}) = \mbox{Tr}_{CE}\left[ V \rho_{A} V^{\ast}\right]$. This is exactly how the complementary maps were originally introduced in the literature.

Finally, the theorem clearly indicates which sub-collection of cptp maps we must consider to validate the previous conjectures. Intuitively, we will work with those classes of channels in which there is a relation between them and their complementary. 

\subsection{Main Results - (Anti)Degradable and Self-degradable Channels}\label{SubSec.ModifiedConjectures}

Technically speaking, the most natural way to leverage the content of thm.~\ref{Thm.CompatibilityAndOrdering} consists of demanding that $\psi_{B|A}$ and its complement $\psi_{E|A}^{c}$ verify a composition rule. This is exactly the key assumption we adopt to finally put together the concepts of divisible maps and compatibility. At a first sight, that class of maps might seem too restrictive or merely an exercise of mathematical abstraction but, as it turns out, these maps have been thoroughly investigated in the literature~\cite{BDHM10,SRZ16,SSWR17}, or more recently in~\cite{CBZ19}, and have a concrete physical meaning, as we will discuss in sub-sec~\ref{SubSec.DiscussionCorrect}. We begin this section formally defining such class of maps (see fig.~\ref{Fig.DefDegradability}):
\begin{definition}[Degradable Maps]
A quantum channel $\psi_{B|A}$ is called \emph{degradable} when there exists another quantum channel $\lambda_{E|B}$ such that
\begin{equation}
\psi_{E|A}^{c}=\lambda_{E|B} \circ \psi_{B|A},
\label{Eq.DefDegradable}
\end{equation}
where $\psi_{E|A}^{c}$ is complementary to $\psi_{B|A}$. On the other hand, when $\psi_{E|A}^{c}$ is degradable, the original map $\psi_{B|A}$ is called \emph{anti-degradable}.
\label{Def.DegradableMaps}
\end{definition}

\textbf{Remark:} Simply put, a degradable channel is a channel that appropriately divides its complementary channel.

\begin{figure}
\includegraphics[scale=0.25]{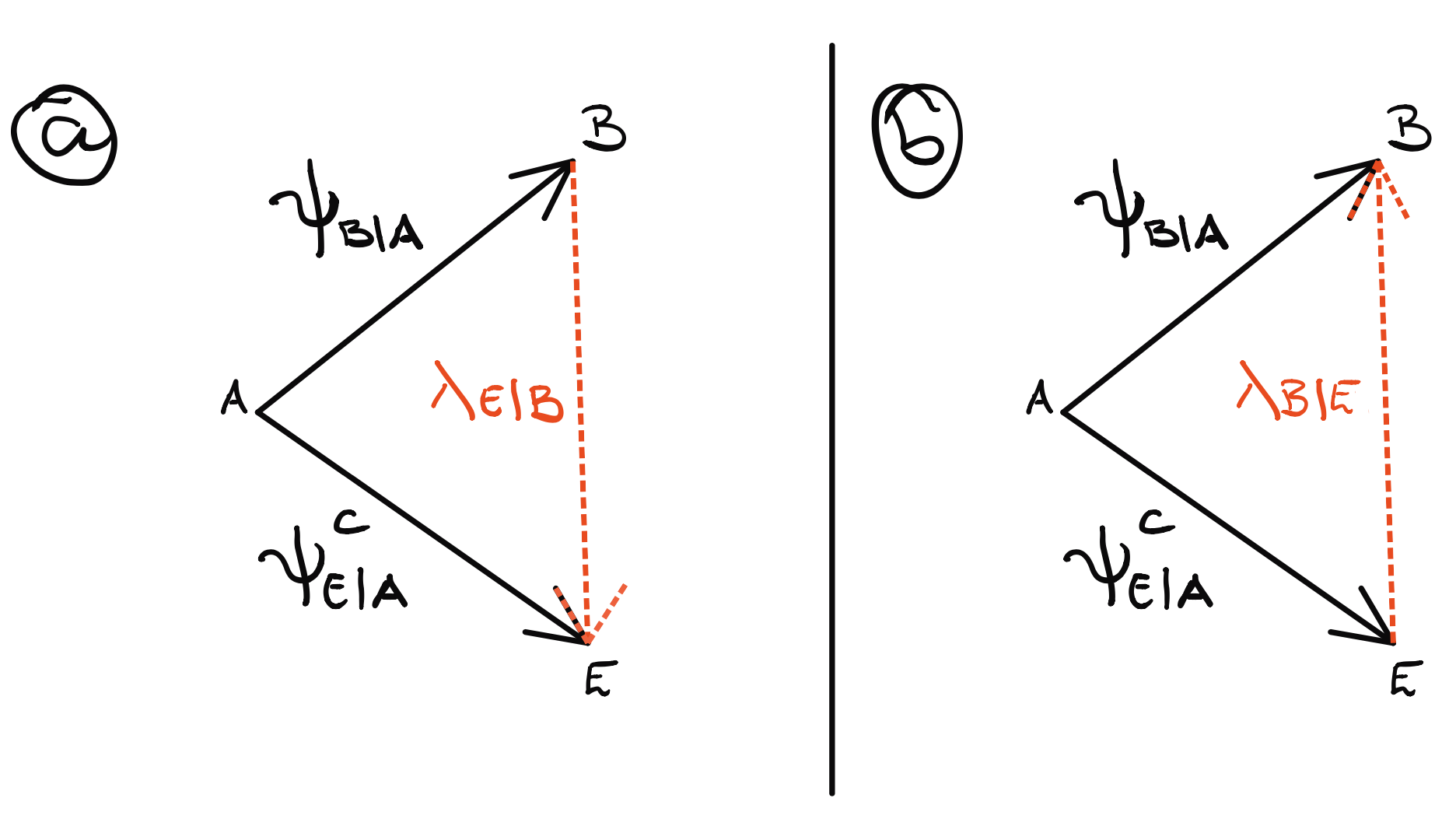}
\caption{Schematic arrow-representation of (a) degradable and (b) anti-degradable maps. On the left, the orange dashed arrow points downwards, from the output space to the environment. On the right, the orange dashed arrow points upwards, from the environment onto the output space. (colours online).\label{Fig.DefDegradability}}
\end{figure}

As we mentioned, the structure of degradable and anti-degradable maps has been deeply explored in the literature before, mainly in the context of quantum information. As an example, it has been shown that, for these channels, the calculation of quantum and private classical capacities (two quantities not even known to be computable in general) simplify massively, to the point that for some degradable channels it is even possible to obtain closed, analytic expressions~\cite{BDHM10,CRS08}. Additionally, even though the set of all degradable channels is not convex, it is known that the set of anti-degradable channels is convex and properly contains all the prepare-and-measure maps~\cite{CRS08}, which implies that anti-degradable maps have positive measure within the set of all channels. In terms of structures, in ref.~\cite{DS05}, the authors have shown that channels having simultaneously diagonalizable Kraus maps are degradable. Additionally, 
approximate notions of degradability and anti-degradability have also been recently explored in the literature ~\cite{SSWR17}.

Albeit being very particular, a special case is when $\lambda_{E|B}$ is the identity channel $\mbox{id}_{E}$. This special case will be useful later in our discussions and we denote it as follows. 

\begin{definition}[Self-Degradable Maps]
A channel $\psi_{B|A}$ is self-degradable (or self-complementary) whenever it coincides with its complementary map. More precisely,
\begin{equation}
\psi_{B|A}=\psi_{E|A}^{c}.
\label{Eq.DefSelfComplementarity}
\end{equation}
\end{definition}
 
One can easily read from eqs.~\eqref{Eq.DefDegradable} and \eqref{Eq.DefSelfComplementarity} that it must be the case that output and environment are the same, in other words $\mathcal{H}_{E}=\mathcal{H}_{B}$. As a simple example, consider also the case where $\mbox{dim}(\mathcal{H}_{A})=\mbox{dim}(\mathcal{H}_{B})=2$. In this case it has been shown that nor the completely depolarizing channel nor any unitary channel is self-complementary ~\cite{SRZ16}. Roughly speaking, we can interpret this argument as to showing that self-complementary channels cannot be too noisy nor too reversible. For qubits, the general structure of these channels  was investigated in ref.~\cite{SRZ16}. Any Kraus decomposition $\{K_1,K_2\}$ of a self-complementary channel must obey:
\begin{align}
&K_1 = 
\begin{pmatrix}
\sin \alpha & 0 \\
0 & \frac{1}{\sqrt{2}}
\end{pmatrix}
\,\,\,
K_2= 
\begin{pmatrix}
0 & \frac{1}{\sqrt{2}} \\
\e^{i\beta}\cos \alpha  & 0
\end{pmatrix}  \mbox{or} \\
&K_1 = 
\begin{pmatrix}
1  & 0 \\
0 & \frac{1}{\sqrt{2}}\sin\alpha
\end{pmatrix}
\,\,\,
K_2= 
\begin{pmatrix}
0 & \frac{1}{\sqrt{2}}\sin\alpha \\
0  & \e^{i\beta}\cos \alpha
\end{pmatrix},  \nonumber
\end{align}
with $\alpha \in [0,\pi]$ and $\beta \in [0,2\pi]$. The most paradigmatic example being the dephasing channel, where $\alpha=0=\beta$.

Given our purposes, we have all the ingredients needed for our main equivalence to hold. Exploiting the extra notions defined in this section, we can finally promote the conjecture of the previous section to a valid theorem. The proof is going to involve a combination of thm.~\ref{Thm.CompatibilityAndOrdering} and def.~\ref{Def.DegradableMaps}.

\begin{thm}
Let $\psi_{B|A}$ and $\phi_{C|A}$ be two quantum channels. The following statements hold true:
\begin{enumerate}
\item [i)] If $\psi_{B|A}$ is degradable and $\psi_{B|A}$ and $\phi_{C|A}$ are compatible, then $\psi_{B|A}$ divides $\phi_{C|A}$.
\item[ii)] If $\psi_{B|A}$ is anti-degradable and $\psi_{B|A}$ divides $\phi_{C|A}$, then $\psi_{B|A}$ and $\phi_{C|A}$ are compatible.
\end{enumerate}
\label{Thm.ModifiedConjectures}
\end{thm}

\begin{proof}

Initially, let us assume that $\psi_{B|A}$ and $\phi_{C|A}$ are compatible. In this case, thm.~\ref{Thm.CompatibilityAndOrdering} ensures that there exists $\theta_{C|E}$ such that $\phi_{C|A}=\theta_{C|E} \circ \psi_{E|A}^{c}$. Now, because $\psi_{B|A}$ is degradable, it is possible to write down its complementary map as $\psi_{E|A}^{c}=\lambda_{E|B} \circ \psi_{B|A}$, with $\lambda_{E|B}$ a cptp map. Hence, $\phi_{C|A}=(\theta_{C|E} \circ \lambda_{E|B}) \circ \psi_{B|A}$ and $\psi$ divides $\phi$.

On the other hand, let us assume now that $\psi_{B|A}$ divides $\phi_{C|A}$. In this case, there exists $\theta_{C|B}$ such that $\phi_{C|A}=\theta_{C|B} \circ \psi_{B|A}$. Now, as we have assumed that $\psi_{B|A}$ is anti-degradable, there exists another cptp map $\lambda_{B|E}$ with $\psi_{B|A}=\lambda_{B|E} \circ \psi_{E|A}^{c}$. Hence, $\phi_{C|A}=(\theta_{C|B} \circ \lambda_{B|E}) \circ \psi_{E|A}^{c}$ and, because of thm.~\ref{Thm.CompatibilityAndOrdering},  $\psi$ is compatible with $\phi$.

\end{proof}

Straightforwardly from thm.~\ref{Thm.ModifiedConjectures}, it follows that with self-degradable channels we can cut out the middle maps $\lambda_{E|B}$ and $\lambda_{B|E}$. In this very particular case compatibility equals divisibility. We formalize it in a corollary.

\begin{corollary}
Let $\psi_{B|A}$ and $\phi_{C|A}$ be two quantum channels. Assume that $\psi_{B|A}$ is self-degradable. The following statements hold true:
\begin{enumerate}
\item[i)] $\psi_{B|A}$ and $\phi_{C|A}$ are compatible.
\item[ii)] $\psi_{B|A}$ divides $\phi_{C|A}$.
\end{enumerate}
\label{Coro.FullEquivalence}
\end{corollary}

\textbf{Remark: } To begin with, note that if we had been working within the usual framework of open quantum dynamics, where it is given a discrete family $\mathcal{F}=\{ \psi_{X_{t}|X_{0}}\}_{t}$ of cptp maps, in order to ensure that compatibility implies divisibility of $\mathcal{F}$,  we would have to assume that every member of that family (but the last one) is degradable. Recursive use of thm~\ref{Thm.ModifiedConjectures} would guarantee that, for being two-by-two compatible, the whole dynamics is Markovian. Analogously, anti-degradability plus the usual definition of divisibility would imply compatibility between consecutive time steps. Finally, despite being rather restrictive, the full equivalence expressed by corollary~\ref{Coro.FullEquivalence} neatly bridges the two concepts. Recall that, for qubits, self-complementarity rules-out the simplest cases of reversible and completely depolarizing channels. 

Overall, we have seen that if we assume degradability or anti-degradability we can see compatibility and divisibility within a closer perspective. Mathematically speaking, although thm.~\ref{Thm.ModifiedConjectures} puts an end to the story, it remains to explain why we have been forced to demand for (anti)degradable maps. Moreover, it is important to reason about the physical meaning of the results.

\subsection{Why (anti)degradable maps?}\label{SubSec.DiscussionCorrect}

Let us recall that thm.~\ref{Thm.ModifiedConjectures} assumes (anti)degradability to establish a connection between divisibility and compatibility. It is natural to ponder over what is the role played by (anti)degradability in these statements. In the next proposition we show that (anti)degradability emerges naturally from assuming compatibility and divisibility.

\begin{proposition}
If $\psi_{B|A}$ and $\phi_{C|A}$ are compatible channels such that $\phi_{C|A}=\theta_{C|B} \circ \psi_{B|A}$, for a given cptp map $\theta_{C|B}$, then $\phi_{C|A}$ is anti-degradable.
\label{Prop.ExplainingAntiDegradability}
\end{proposition}

\begin{proof}
To begin with, because $\psi_{B|A}$ and $\phi_{C|A}$ are compatible, there is another cptp map $\theta_{B|E}$ such that $\psi_{B|A} = \theta_{B|E} \circ \phi_{E|A}^{c}$. Now, as $\phi_{C|A}=\theta_{C|B} \circ \psi_{B|A}$, we conclude:
\begin{align}
\phi_{C|A}&=\theta_{C|B} \circ \psi_{B|A} \nonumber \\
& = \theta_{C|B} \circ \theta_{B|E} \circ \phi_{E|A}^{c}.
\end{align}
\end{proof}

Notice that, in prop.~\ref{Prop.ExplainingAntiDegradability}, not only we are getting anti-degradability out of the result, rather than plain degradability, but also note that it is $\phi$ which ends up being anti-degradable. Remember that in thm.~\ref{Thm.ModifiedConjectures} we had to demand either degradability or anti-degradability of $\psi$ to obtain divisibility or compatibility, respectively. 

An alternative version of the above proposition is also possible. Instead of requiring compatibility for the channels $\psi_{B|A}$ and $\phi_{C|A}$, we could have asked for compatibility between the conjugate version of the channels. Similarly to prop.~\ref{Prop.ExplainingAntiDegradability}'s proof, we would have obtained that $\psi$ must be degradable. 
Other variants can also be obtained if we move around the super-index indicating the conjugate channel.

We conclude this section discussing the physical intuition behind our results. To begin with, consider the two diagrams displayed in fig.~\ref{Fig.DefDegradability}. Roughly speaking, they show that degradable maps are those in which the evolution of the environment can be fully determined if we monitor the system. To be a bit more precise, the environment evolved states are fully encoded in the system, in the sense that one can recover any environment evolved state by a suitable, but fixed, physical operation that can be done on the system alone. For anti-degradable maps, a similar picture holds, but with the environment and system papers reversed. Self-degradable maps are, therefore, those in which system and environment are on equal footing, to the extent that we can imagine to do process tomography on the system to obtain the information on the evolution of the environment and vice-versa. Let us focus on this class of maps. 

The content of thm.~\ref{Thm.CompatibilityAndOrdering} is that if $\psi_{B|A}$ and $\phi_{C|A}$ are compatible, the compatibilizer $\theta_{BC|A}$ also stores in itself the information about the environment described by $\psi^{c}_{E|A}$, and vice versa. It is as if the black-box of fig.~\ref{Fig.DefCompatibility} were taken so large that it also encompasses the environment of a Lindblad-Stinespring dilation for $\psi_{B|A}$, and so that $\psi^{c}_{E|A}$ contains the information of $\phi_{C|A}$. In fact,
\begin{equation}\phi_{C|A}=\mbox{Tr}_{B}(\theta_{BC|A})=\theta_{C|E} \circ \psi^{c}_{E|A}.\end{equation} Now, self-degradability says that we can bring $\psi_{B|A}$ back into the game, as environment and system are on equal footing, $\psi^{c}_{E|A}=\psi_{B|A}$. As a consequence, $\psi_{B|A}$ ends up dividing $\phi_{C|A}$, thus containing information about it. On the other hand, assuming that $\psi_{B|A}$ divides $\phi_{C|A}$, with the former being self-degradable, it is possible to write $\phi_{C|A}$ in terms of the environment in $\psi^{c}_{E|A}$. In this way, one can find an appropriate large box $\theta_{BC|A}$ involving the environment $E$ for $\psi_{B|A}$, such that it acts as the compatibilizer for $\psi_{B|A}$ and $\phi_{C|A}$. In brief, self-degradability bridges in both directions the evolution of environment and system (up to a function), and this is the very key point in equating compatibility and divisibility, as we can make the compatibilizer large enough to accommodate the environment.

\begin{figure*}
\includegraphics[scale=0.4]{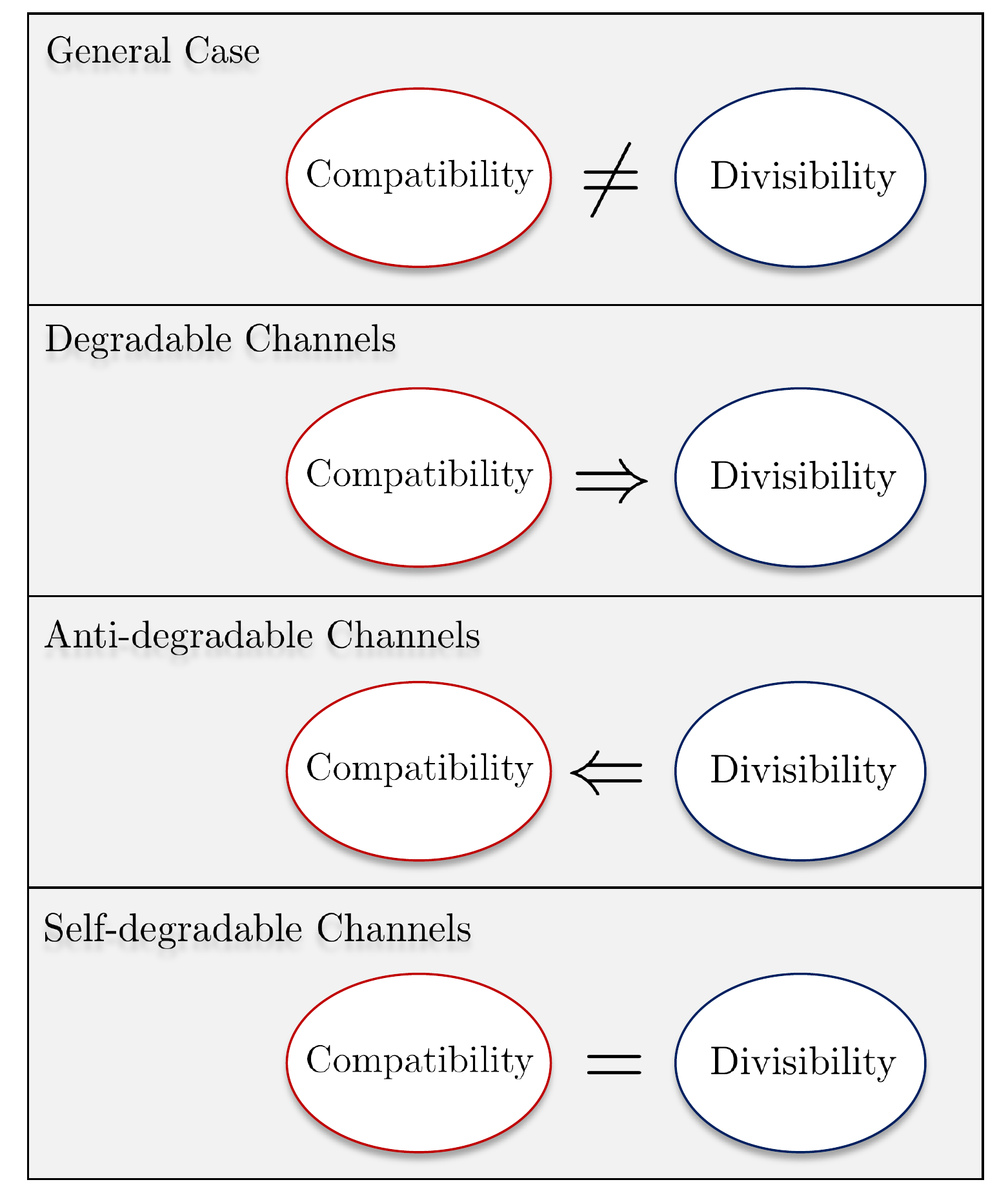}
\caption{Concise summary of results. With no restrictions on the class of channels, the difference between the parallel character of compatibility and the sequential character of divisibility prevails and the two notions are not related. If the channel $\psi_{B|A}$ is degradable and $\psi_{B|A}$ and $\phi_{C|A}$  are compatible, then $\psi_{B|A}$ divides $\phi_{C|A}$. If the channel $\psi_{B|A}$ is anti-degradable and $\psi_{B|A}$ divides $\phi_{C|A}$, then $\psi_{B|A}$ and $\phi_{C|A}$ are compatible. If $\psi_{B|A}$ is self-degradable, it is equivalent to state that $\psi_{B|A}$ divides $\phi_{C|A}$ and that $\psi_{B|A}$ and $\phi_{C|A}$ are compatible.}
\label{Fig.Summary}
\end{figure*}

\section{Conclusions}\label{Sec.Conclusion}

In this work we have established a connection between compatibility and divisibility of cptp maps (summarized in fig.~\ref{Fig.Summary}). We have shown that for self-degradable maps there is a direct equivalence between the two concepts: compatibility implies divisibility and vice versa. Slightly relaxing this hypothesis, compatible and degradable maps are divisible. Furthermore, divisible and anti-degradable maps are compatible. These implications and equivalences do not hold true if we remove (anti)degradability from the body of hypothesis, as we have managed to show that neither compatibility directly implies divisibility nor that divisibility directly implies compatibility in general. Interestingly, the latter is a consequence of the non-broadcasting theorem.

Apart from some cases, it is known that determining whether or not a certain map, or a family of maps, is (anti)degradable is not a trivial task. However, if one knows beforehand that degradability holds, thm~\ref{Thm.ModifiedConjectures} can be used to determine divisibility with a number of SDPs that scales with the size of the family under consideration. Recall that from an information theoretical perspective, the authors of ref.~\cite{DNC21} have come up with an algorithm that must implement an infinite number of SDPs to determine the divisibility of a single pair of maps. With extra information, our results might be useful in this case.

There are, of course, open questions that originate from this work. A natural one is whether the class of self-degradable channels identifies the class of channels for which compatibility equals divisibility. This question alone motivates further research on this special class of channels. To name other open challenges, consider for instance the argument of example 3. There we constructed a situation where the maps are compatible and not-divisible. Our entire construction was based upon the intuition that conditionally independent maps are not divisible. Although this makes sense, and although it works for that example, we still lack a general mathematical proof of this statement. We can push this direction even further: it is still unknown whether or not compatible maps possess at least one compatibilizer respecting conditional independence. This will also be explored in a future work.

Finally, we recognize the limitations of our work. The broad conjectures 1 and 2 do not hold true, and the theorem~\ref{Thm.ModifiedConjectures} and Corollary~\ref{Coro.FullEquivalence} are quite restrictive. Nonetheless, we emphasize how our work puts together two concepts that, up to this day, had always been thought of being separate. Each of these concepts has its own toolbox, and we believe that our work can help to pave-down a new road allowing for an interplay  of these toolboxes. We have also been able to identify open questions that should be explored in future research. 

\begin{acknowledgments}
Cristhiano Duarte wishes to thank Nadja Bernardes, Fernando de Melo, Matt Leifer, Zolt\'{a}n Zimbor\'{a}s and Marcelo Terra Cunha for all the extremely useful discussions. Part of this work was carried out in a transition period and CD thanks Rafael Chaves and the directors of the International Institute of Physics for all the support. R. Drumond acknowledges the support from Brazilian agencies Coordena\c{c}\~ao de Aperfei\c{c}oamento de Pessoal de N\'{i}vel Superior-Finance Code 001, and FAPEMIG. This paper is a result of the Brazilian National Institute of Science and Technology on Quantum Information. L. Catani acknowledges the support by the Fetzer Franklin Fund of the John E.\ Fetzer Memorial Trust and by the Army Research Office (ARO) (Grant No. W911NF-18-1-0178). This research was supported by the Fetzer Franklin Fund of the John E.\ Fetzer Memorial Trust and by grant number FQXi-RFP-IPW-1905 from the Foundational Questions Institute and Fetzer Franklin Fund, a donor advised fund of Silicon Alley Community Foundation.  
 
\end{acknowledgments}

\bibliography{list_of_references}
\end{document}